\documentclass[envcountsect,envcountsame]{llncs}
\usepackage{amsmath,amssymb,bbm,graphicx}
\usepackage{enumerate,verbatim}
\usepackage{tikz}
\usepackage{epic,eepic,phgraphs}
\usepackage{xspace}
\tikzstyle{every picture} = [>=latex]
\spnewtheorem{algorithm}[theorem]{Algorithm}{\bf}{\rm}

\def\ca#1{{\cal#1}}
\def\pplusoid{\mathop{\mbox{$\otimes$}}}
\let\ljoin\pplusoid
\newcommand\dbar[1]{\widetilde#1}
\let\sem\setminus

\def\pplusx#1#2{\,{[#1\,|\>#2]}\,}
\def\pplus#1#2{\,\ppluso{#1\,|\>#2}\,}
\def\ppluso#1{\mathop{\mbox{$\otimes[#1]\,$}}}
\def\prebox#1{\mathop{\mbox{\sl#1}}}
\def\GFii{\mathrm{GF}(2)}
\def\mx#1{\mbox{\boldmath$#1$}}
\newcommand\rwd{\operatorname{rwd}}
\newcommand\cwd{\operatorname{cwd}}
\newcommand\tw{\operatorname{tw}}
\newcommand\bw{\operatorname{bw}}

\def\atable{\mbox{\sl Table}}

\newcommand{\nb}[1]{}

\newcommand{\maxsat}{\textsc{Max-SAT}\xspace}
\newcommand{\numsat}{\textsc{\#SAT}\xspace}
\newcommand{\sat}{\textsc{SAT}\xspace}

\begin{document}

\title{Better algorithms for satisfiability problems\\ for formulas of bounded
rank-width}
\author{Robert~Ganian \and\ Petr~Hlin\v{e}n\'y \and\ Jan Obdr\v z\'alek
\thanks{This research has been supported by the Czech research grant
GA 201/09/J021 and by the 
Institute for Theoretical Computer Science ITI, project 1M0545.}
}
\institute{
	Faculty of Informatics, Masaryk University\\
	Botanick\'a 68a, Brno, Czech Republic\\[.5ex]
\email{\{xganian1,hlineny,obdrzalek\}@fi.muni.cz}
}

\maketitle

\begin{abstract}
We provide a parameterized polynomial algorithm for the propositional model counting
problem \numsat, the runtime of which is single-exponential in the rank-width
of a formula.
Previously, analogous algorithms have been known --
e.g.~[Fischer, Makowsky, and Ravve] -- with a single-exponential dependency on the
clique-width of a formula.
Our algorithm thus presents an exponential runtime improvement
(since clique-width reaches up to exponentially higher values than
rank-width), and can be of practical interest for small values of
rank-width.
We also provide an algorithm for the \maxsat problem along the same lines.

\medskip\keywordname\
propositional model counting; satisfiability; rank-width; clique-width; parameterized complexity.
\end{abstract}

\section{Introduction}

The satisfiability problem for Boolean formulas in conjunctive normal form
(known as \sat) has been of great practical and theoretical interest for
decades. It is known to be NP-complete, even though many instances are practically solvable
using the various \sat-solvers. We focus on two
well-known generalizations of this problem, namely \numsat and \maxsat. In
\numsat -- otherwise known as the {\em propositional model counting} problem,
the goal is to compute the number of satisfying truth assignments
for an input formula $\phi$, whereas in \maxsat we ask for the maximum
number of simultaneously satisfiable clauses of $\phi$. It is known that
computing \numsat is \#P-hard \cite{Val79} and that \maxsat is already
NP-hard to approximate within some constant \cite{AS98}.

In light of these hardness results, we may ask what happens if we
restrict ourselves to some
subclass of inputs. 
The parameterized algorithmics approach is suitable in such a case. Let
$k$ be some parameter associated with the input instance. 
Such a decision problem is said to be \emph{fixed-parameter
  tractable (FPT)} if it is solvable in time $\ca O(n^p\cdot f(k))$ for some
constant $p$ and a computable function $f$. So the running time is
polynomial in the size $n$ of the input, but can be e.g. exponential in the 
parameter $k$. Obviously the specific form of $f$ plays an important
role in practical applicability of any such algorithm -- while FPT algorithms
with single-exponential $f$ can be feasible for non-trivial values of the
parameter, a double-exponential $f$ would make the algorithm impractical for
almost all values of $k$.

But what are suitable parameters for satisfiability problems?
In the particular case of \maxsat, one can consider the desired 
number of satisfied or unsatisfied clauses as a parameter of the input,
such as in \cite{ck04,ro09}, respectively.
Although, such approach is not at all suitable for \numsat which is
our prime interest in this paper.

Another approach used for instance by Fischer, Makowsky and Ravve \cite{FMR07}
represents the formula $\phi$ as a formula graph $F_\phi$ (nodes of
which are the clauses and variables of $\phi$, see
Definition~\ref{def:formgraph}),
and exploits the fact that for graphs there are many known (and intensively studied)
so called \emph{width parameters}. 
In~\cite{FMR07} the
authors presented FPT algorithms for the \numsat problem in the case of two
well known width parameters -- {\em tree-width} and {\em clique-width}.
A similar idea was used by Georgiou and Papakonstantinou \cite{GP08}
also for the \maxsat problem and by Samer and Szeider \cite{SS10}
for \numsat. 

The latter algorithms work by dynamic programming on tree-like decompositions
related to the width parameters (tree-decompositions and
clique-decompositions -- often called $k$-expressions -- in the cases
above). However, there is the separate issue of the complexity of computing
the width of the formula graph and its decomposition. In the case of
tree-width this can be done in FPT \cite{Bod97}.  
For the much more general clique-width (every graph
of bounded tree-width also has bounded clique-width, while the converse does
not hold) there exist no such algorithms and we rely on approximations or an
oracle. In~\cite{SS10} the authors made the following statement on this issue:

\begin{quote}
  A single-exponential algorithm (for \numsat)
  is due to Fisher, Makowsky, and Ravve~\cite{FMR07}.
  However, both algorithms rely on clique-width approximation
  algorithms. The known polynomial-time algorithms for that purpose admit
  an exponential approximation error \cite{HO08} and are of limited
  practical value.
\end{quote}

The exponential approximation error mentioned in this statement results by bounding
the clique-width by a another, fairly new, width parameter called
\emph{rank-width} (Definition~\ref{def:rankwidth}). 
Rank-width is bounded if and only if clique-width is
bounded, but its value can be exponentially lower than that of clique-width
(Theorem~\ref{thm:rwinequalities}\,a,b).
And since clique-width generalizes tree-width, so
does rank-width (Theorem~\ref{thm:rwinequalities}\,c).
Moreover, for rank-width we can efficiently compute the
related decomposition (Theorem~\ref{thm:rwalg}), which is in stark contrast to the case for
clique-width. Therefore an algorithm which is
linear in the formula size and {\em single-exponential in its rank-width}
challenges the claim quoted above, and can be of real practical value. In
this paper we present such algorithms for the problems \numsat and
\maxsat. More precisely we prove the following two results:

\begin{theorem}
\label{thm:main}
Both the \numsat and \maxsat problems have FPT algorithms running in time
$$\ca O( t^3\cdot 2^{3t(t+1)/2}\cdot|\phi|)$$
where $t$ is the rank-width of the input instance (CNF formula) $\phi$.
\end{theorem}
We refer to further Theorems \ref{thm:signedrwalg}, \ref{thm:nsat}
and~\ref{thm:maxsat} for details.

Note that our results present an {\em actual exponential runtime
improvement} in the
parameter over any algorithm utilizing the clique-width measure,
including aforementioned \cite{FMR07}.
This is since any parameterized algorithm $\ca A$ for a \sat problem has to depend at
least exponentially on the clique-width of a formula
(unless the exponential time hypothesis fails),
and considering typical instances $\phi$ as from Proposition~\ref{pro:srwinequalities}\,b,
such an algorithm $\ca A$ then runs in time doubly-exponential in the
rank-width of~$\phi$.

As for potential practical usefulness of Theorem~\ref{thm:main}, 
note that there are no ``large constants'' hidden in the $\ca O$-notation.
One may also ask whether there are any interesting classes of graphs of low rank-width. 
The answer is a resounding YES, since already for $t=1$ 
we obtain the very rich class of distance-hereditary graphs. 
Rank-width indeed is a very general graph width measure.

\medskip

The approach we use to prove both parts of Theorem~\ref{thm:main}
quite naturally extends the clever and skilled new algebraic methods 
of designing parameterized algorithms for graphs of bounded rank-width,
e.g.~\cite{CK07,BTV10,GH09}, to the area of \sat problems.
Yet, this is not a trivial extension---we remark that 
a straightforward translation of the algorithm
of~\cite{FMR07} from clique-width expressions to rank-decompositions
(which is easily possible) would result just in a double-exponential runtime
dependency on the rank-width.

The rest of the paper is organized as follows: In
Section~\ref{sec:rankwidth} we present the rank-width measure
and some related technical considerations.
This is applied to signed graphs of \sat formulas.
Section~\ref{sec:numsat} then presents our FPT algorithm for
the \numsat problem (Theorem~\ref{thm:nsat} and Algorithm~\ref{alg:nsat}), 
and Section~\ref{sec:maxsat} the similar algorithm for \maxsat
(Theorem~\ref{thm:maxsat}). 
We conclude with some related observations.

\section{Overview of the rank-width measure}
\label{sec:rankwidth}

Graph rank-width~\cite{OS06}, the core concept of our paper, is not so well known,
and hence we give a detailed technical introduction to this concept
and its application to CNF formulas in this section.
Readers familiar with the concept of rank-width (and parse trees for rank-width) may proceed directly to
Section~\ref{sec:signed_graphs}.

\subsection{Branch-width and rank-width}

The usual way of defining rank-width is via the {branch-width} of the 
{cut-rank} function (Definition~\ref{def:rankwidth}).
A set function $f:2^M\rightarrow \mathbb{Z}$ is \emph{symmetric}
if $f(X)=f(M\setminus X)$ for all $X\subseteq M$.
A tree is \emph{subcubic} if all its nodes have degree at most~$3$.
For a symmetric function $f:2^M\rightarrow \mathbb{Z}$ on
a finite ground set~$M$,
the branch-width of $f$ is defined as follows:

A \emph{branch-decomposition} of $f$ is a pair $(T,\mu)$ of 
a subcubic tree $T$ and a bijective function 
$\mu:M\rightarrow \{t: \text{$t$ is a leaf of $T$}\}$.
For an edge $e$ of $T$, the connected components of $T\setminus e$
induce  a bipartition $(X,Y)$ of the set of leaves of $T$.
The \emph{width} of an edge $e$ of a
branch-decomposition $(T,\mu)$ is $f(\mu^{-1} (X))$. 
The \emph{width} of $(T,\mu)$ is the maximum width over all edges of $T$.
The \emph{branch-width}  of $f$ is the minimum of the  width of all
branch-decompositions of $f$. 

\begin{definition}[Rank-width \cite{OS06}]\rm
\label{def:rankwidth}
For a simple graph $G$ and $U,W\subseteq V(G)$, let $\mx A_G[U,W]$ be the matrix
defined over the two-element field $\mathrm{GF}(2)$ as follows:
the entry $a_{u,w}$, $u\in U$ and $w\in W$, of $\mx A_G[U,W]$ 
is $1$ if and only if $uw$ is an edge of~$G$.
The {\em cut-rank} function $\rho_G(U)=\rho_G(W)$ then equals the rank of 
$\mx A_G[U,W]$ over $\mathrm{GF}(2)$ where $W=V(G)\sem U$.
A \emph{rank-decomposition} (see Figure~\ref{fig:rdecC5})
and \emph{rank-width} of a graph $G$ 
is the branch-decomposition and branch-width of the cut-rank function
$\rho_G$ of~$G$ on $M=V(G)$, respectively.
\end{definition}

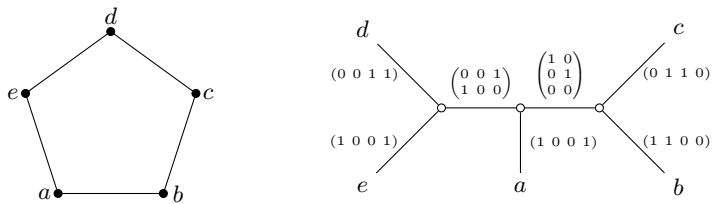
\begin{figure}[ht]
{\hfill
\begin{tikzpicture}[scale=0.7]
\tikzstyle{every node}=[draw, shape=circle, minimum size=3pt,inner sep=0pt, fill=black]

    \draw (0,0) node (a) [label=left:$a$] {}
        -- ++(0:2.0cm) node (b) [label=right:$b$] {}
        -- ++(72:2.0cm) node (c) [label=right:$c$] {}
        -- ++(144:2.0cm) node (d) [label=above:$d$] {}
        -- ++(216:2.0cm) node (e) [label=left:$e$] {}
        -- (a); 
\end{tikzpicture}
\qquad\qquad
\begin{tikzpicture}[x=1.5cm,y=1.5cm, scale=0.7]
\tikzstyle{every node}=[]
\tikzstyle{empty}=[draw, shape=circle, minimum size=3pt,inner sep=0pt, fill=white]
\node (e) at (0,0) {$e$};
\node (d) at (0,2) {$d$};
\node (ed) at (1,1) [empty] {};
\node (mid) at (2,1)[empty] {};
\node (bc) at (3,1) [empty] {};
\node (b) at (4,0) {$b$};
\node (c) at (4,2) {$c$};
\node (a) at (2,0) {$a$};

\draw (e) to node[left] {\tiny$(1\ 0\ 0\ 1)$} (ed);
\draw (d) to node[left] {\tiny$(0\ 0\ 1\ 1)$} (ed);
\draw (ed) to node[above] {\tiny$\begin{pmatrix}0&0&1\cr 1&0&0\end{pmatrix}$} (mid);
\draw (mid) to node[above] {\tiny$\begin{pmatrix}1&0\cr 0&1\cr 0&0\end{pmatrix}$} (bc);
\draw (bc) to node[right] {\tiny$(1\ 1\ 0\ 0)$} (b);
\draw (bc) to node[right] {\tiny$(0\ 1\ 1\ 0)$} (c);
\draw (a) to node[right] {\tiny$(1\ 0\ 0\ 1)$} (mid);
\end{tikzpicture}
\hfill}
\medskip
\caption{A rank-decomposition of the graph cycle $C_5$,
	showing the matrices involved in evaluation of its cut-rank function
	on the edges of the decomposition.}
\label{fig:rdecC5}
\end{figure}

As already mentioned in the introduction, rank-width is closely related to clique-width
and more general than better known tree-width.
Indeed:

\begin{theorem}
\label{thm:rwinequalities}
Let $G$ be a simple graph, and $\tw(G)$, $\bw(G)$, $\cwd(G)$, $\rwd(G)$
denote in this order the tree-width, branch-width, clique-width, and
rank-width of $G$.
Then the following hold
\begin{enumerate}[a)]\vspace*{-3pt}
\item\cite{OS06}~
$\rwd(G)\leq \cwd(G) \leq 2^{\rwd(G)+1}-1$,
\item\cite{CR05}~~
the clique-width $\cwd(G)$ can reach up to $2^{\rwd(G)/2-1}$,
\item\cite{Oum08}~
$\rwd(G)\leq\bw(G) \leq\tw(G)+1$,
\item{[folklore]}
$\tw(G)$ cannot be bounded from above by $\rwd(G)$,
e.g.\ the complete graphs have rank-width $1$ while their tree-width is
unbounded,
\item\cite{Oum05}~
$\rwd(G)=1$ if and only if $G$ is a distance-hereditary graph.
\end{enumerate}
\end{theorem}

Although rank-width and clique-width are ``tied together'' (a),
one of the crucial advantages of rank-width is its parameterized
tractability
(on the other hand, it is not known how to efficiently test
$\cwd(G)\leq k$ for $k>3$):

\begin{theorem}[\cite{HO08}]
\label{thm:rwalg}~
There is an FPT algorithm
that, for a fixed parameter $t$ and a given graph $G$,
either finds a rank-decomposition of $G$ of width at most~$t$
or confirms that the rank-width of $G$ is more than~$t$.
\end{theorem}

\subsection{Labeling parse trees for rank-width}
\label{sub:parsetr}

Unlike for tree-width and clique-width, the standard definition of
rank-decom\-positions 
is not suitable for the immediate design of efficient algorithms.
To this end, closely following Courcelle and Kant\'e~\cite{CK07},
we have introduced so-called {\em labeling parse trees} \cite{GH09} 
(Definition~\ref{df:parsetree} and Figure~\ref{fig:parseC5}) --
a powerful formalism for dynamic programming design on graphs of bounded rank-width.
The basic idea is to transform rank-decompositions into suitable parse trees
and have algorithms use them instead of the decomposition.

A {\em$t$-labeling} of a graph is a mapping \mbox{$lab:V(G)\to2^{[t]}$} where
$[t]=\{1,2,\dots,t\}$ is the set of {\em labels}.
Having a graph $G$ with an associated $t$-labeling $lab$,
we refer to the pair $(G,lab)$ as to a {\em$t$-labeled graph} and use notation~$\bar G$.
We will often view a $t$-labeling of $G$ equivalently as a mapping $V(G)\to \GFii^t$ 
to the {\em binary vector space} of dimension~$t$, where $\GFii$ is the
two-element finite field.

\begin{definition}\rm
\label{def:ljoin}
Considering $t$-labeled graphs $\bar G_1=(G_1,lab^1)$ and
$\bar G_2=(G_2,lab^2)$,
a {\em$t$-labeling join} $\bar G_1\pplusoid\bar G_2$ is defined on the
disjoint union of $G_1$ and $G_2$ by adding
all edges $(u,v)$ such that $|lab^1(u)\cap lab^2(v)|$ is odd,
where $u\in V(G_1),v\in V(G_2)$.
(Alternatively,
$\{u,v\}$ is an edge of $\bar G_1\ljoin\bar G_2$ if and only if
$lab^1(u)\cdot lab^2(v)=1$ over $\mathrm{GF}(2)$.)
The resulting graph is unlabeled.
\end{definition}

A {\em$t$-relabeling} is a mapping $f:[t]\to2^{[t]}$.
In linear algebra terms, a $t$-rela\-beling $f$ is in a natural one-to-one correspondence
with a {\em linear transformation} $f:\GFii^t\to\GFii^t$,
i.e.\ a $t\times t$ binary matrix $\mx R_f$.
For a $t$-labeled graph $\bar G=(G,lab)$ we define $f(\bar G)$ as the same graph
with a vertex $t$-labeling $lab'=f\circ lab$.
Here $f\circ lab$ stands for the linear transformation $f$ applied to the
labeling $lab$, or equivalently $lab'=lab\times\mx R_f$ as matrix multiplication
over $\GFii^t$.

\begin{definition}[Labeling parse tree \cite{CK07,GH09}]\rm
\label{df:parsetree}
Let $\odot$ be a nullary operator creating a single new graph vertex of
label~$\{1\}$.
For $t$-relabelings $f_1,f_2,g:[t]\to2^{[t]}$, let $\pplus g{f_1,f_2}$ be a
binary operator\,---\,called {\em$t$-labeling composition}%
\,---\,over pairs of $t$-labeled graphs
$\bar G_1=(G_1,lab^1)$ and $\bar G_2=(G_2,lab^2)$ defined
(cf.~\ref{def:ljoin})
$$
        \bar G_1 \pplus g{f_1,f_2} \bar G_2 ~=~ \bar H
        \,=\, \big(\bar G_1\ljoin g(\bar G_2),\, lab\big)
$$
where the new labeling is $lab(v)=f_i\circ lab^i(v)$ for $v\in V(G_i)$,
$i=1,2$.
In other words, $\{u,v\}\in E(H)$ where $u\in V(G_1)$, $v\in V(G_2)$,
if and only if
\mbox{$lab^1(u)\times\mx R_g^T\times lab^2(v)^T=1$} over $\mathrm{GF}(2)$
(cf.\ Courcelle and Kant\'e~\cite{CK07}).

\smallskip
A {\em $t$-labeling parse tree} $T$
is a finite rooted ordered subcubic tree
(with the root degree at most $2$) such that
\begin{itemize}\vspace{-3pt}
\item all leaves of $T$ contain the $\odot$ symbol, and
\item each internal node of $T$ contains one of the $t$-labeling composition
symbols.
\end{itemize}\vspace{-3pt}
A parse tree $T$ then {\em generates} (parses) the graph $G$ which is
obtained
by successive leaves-to-root applications of the operators in the nodes
of~$T$.
\end{definition}
\vspace*{-1ex}

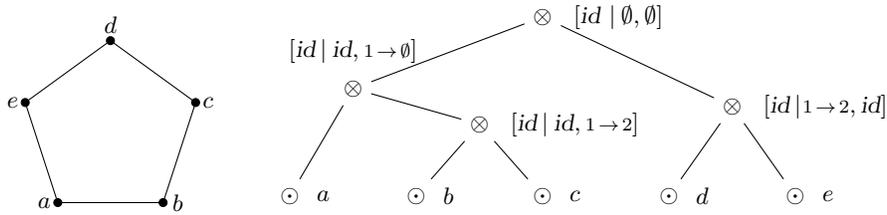
\begin{figure}[ht]
{\hfill
\begin{tikzpicture}[scale=0.7]
\tikzstyle{every node}=[draw, shape=circle, minimum size=3pt,inner sep=0pt, fill=black]
\draw (0,0) node (a) [label=left:$a$] {}
        -- ++(0:2.0cm) node (b) [label=right:$b$] {}
        -- ++(72:2.0cm) node (c) [label=right:$c$] {}
        -- ++(144:2.0cm) node (d) [label=above:$d$] {}
        -- ++(216:2.0cm) node (e) [label=left:$e$] {}
        -- (a); 
\end{tikzpicture}
\quad\quad
\begin{tikzpicture}[x=1.2cm,y=1.7cm, scale=0.7]
\tikzstyle{every node}=[]
\tikzstyle{empty}=[draw, shape=circle, minimum size=3pt,inner sep=0pt, fill=white]
\node (a) at (0,0) [label=right:$a$] {$\odot$} ;
\node (b) at (2,0) [label=right:$b$] {$\odot$};
\node (c) at (4,0) [label=right:$c$] {$\odot$};
\node (d) at (6,0) [label=right:$d$] {$\odot$};
\node (e) at (8,0) [label=right:$e$] {$\odot$};
\node (abc) at (1,1.2) [label=above:$\pplusx{\prebox{id}\!}{\!\prebox{id},\mbox{\scriptsize$1\!\to\!\emptyset$}}$] {$\otimes$};
\node (bc) at (3,0.8) [label=right:$\pplusx{\prebox{id}\!}{\!\prebox{id},\mbox{\scriptsize$1\!\to\!2$}}$] {$\otimes$};
\node (de) at (7,1)   [label=right:$\pplusx{\prebox{id}\!}{\!\mbox{\scriptsize$1\!\to\!2$},\prebox{id}}$] {$\otimes$};
\node (abcde) at (4,2)[label=right:$\pplusx{\prebox{id}}{\emptyset,\emptyset}$] {$\otimes$};
\draw (a) to (abc);
\draw (b) to (bc);
\draw (c) to (bc);
\draw (d) to (de);
\draw (e) to (de);
\draw (bc) to (abc);
\draw (abc) to (abcde);
\draw (de) to (abcde);
\end{tikzpicture}
\hfill}
\caption{An example of a $2$-labeling parse tree which generates a cycle
	$C_5$, with symbolic relabelings at the nodes
	($\prebox{id}$~denotes the relabeling preserving all labels,
	and $\emptyset$ is the relabeling ``forgetting'' all labels).}
\label{fig:parseC5}
\end{figure}

\begin{figure}[ht]
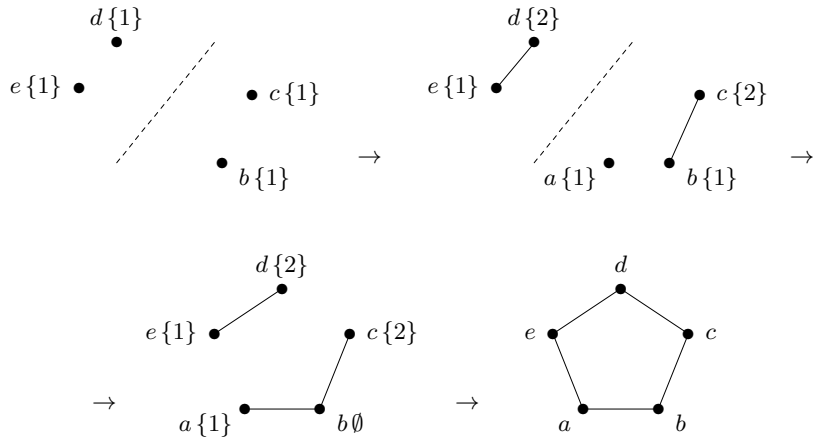

\begin{center}\small
\grpicture(200,100)(-40,0){.2}
\grvert B(70,0),C(90,45),D(-0,80),E(-25,50)
\grnovert x(0,0),y(65,80)\gredged xy
\grnamevert rbB{$b\,\{1\}$},rcC{$c\,\{1\}$},ctD{$d\,\{1\}$},lcE{$e\,\{1\}$}
\endgrpicture
$\to$\qquad\qquad\quad
\grpicture(190,100)(-40,0){.2}
\grvert A(30,0),B(70,0),C(90,45),D(-20,80),E(-45,50)
\grnovert x(-20,0),y(45,80)\gredged xy
\gredge BC,DE
\grnamevert lbA{$a\,\{1\}$},rbB{$b\,\{1\}$},rcC{$c\,\{2\}$},ctD{$d\,\{2\}$},lcE{$e\,\{1\}$}
\endgrpicture$\to$
\par\bigskip\bigskip
$\to$\qquad\qquad
\grpicture(160,120){.2}
\grvert A(20,0),B(70,0),C(90,50),D(45,80),E(0,50)
\gredge AB,BC,DE
\grnamevert lbA{$a\,\{1\}$},rbB{$b\,\emptyset$},rcC{$c\,\{2\}$},ctD{$d\,\{2\}$},lcE{$e\,\{1\}$}
\endgrpicture
$\to$\qquad\quad
\grpicture(150,120){.2}
\grvert A(20,0),B(70,0),C(90,50),D(45,80),E(0,50)
\gredge AB,BC,CD,DE,EA
\grnamevert lbA{$a$},rbB{$b$},rcC{$c$},ctD{$d$},lcE{$e$}
\endgrpicture
\end{center}
\caption{``Bottom-up'' generation of $C_5$ by the parse tree from
	\figurename~\ref{fig:parseC5}.}
\label{fig:parsegenC5}
\end{figure}

See Figures~\ref{fig:parseC5} and~\ref{fig:parsegenC5}.
The crucial statement is that rank-decompositions are exactly equivalent to
labeling parse trees:

\begin{theorem}[\cite{CK07,GH09}]
\label{thm:parse}~
A graph $G$ has rank-width at most $t$ if and only if
(some labeling of) $G$ can be generated by a $t$-labeling parse tree.
Furthermore, a width-$t$ rank-decomposition of $G$ can be transformed into
a $t$-labeling parse tree on $\Theta(|V(G)|)$ nodes in time~$\ca O(t^2\cdot|V(G)|^2)$.
\end{theorem}

\subsection{Signed graphs and rank-width of CNF formulas}
\label{sec:signed_graphs}

Although there are several methods for converting formulas to graphs, the most
common and perhaps most natural approach uses so-called {\em signed graphs} 
(e.g. \cite{FMR07,SS10,gksy09}).
A signed graph is a graph $G$ with two edge sets $E^+(G)$ and $E^-(G)$.
We refer to its respective positive and negative subgraphs as to
$G^+$ and $G^-$.
Notice that $G^+$ and $G^-$ are edge-disjoint and $G=G^+\cup G^-$.

\begin{definition}\rm
\label{def:formgraph}
The signed graph $F_\phi$ of a CNF formula $\phi$ is defined as follows:
\begin{itemize}\vspace*{-3pt}
\item      
$V(F_\phi)=W\cup C$ where $W$ is the set of variables occurring in $\phi$ and 
$C$ is the set of clauses of $\phi$.
\item
For $w\in W$ and $c\in C$, it is $wc\in E^+(F_\phi)$ iff the literal `$w$'
occurs in~$c$.
\item
For $w\in W$ and $c\in C$, it is $wc\in E^-(F_\phi)$ iff the literal `$\neg w$'
occurs in~$c$.
\end{itemize} 
\end{definition}

Since signed graphs have two distinct edge sets, the definition of
rank-width needs to be modified to reflect this.
It should be noted that simply using two separate, independent decompositions
would not work -- the bottom-up dynamic programming algorithm we are going to
use will need information from both edge sets at every node to work properly.
Instead, one may define, analogically to Definition~\ref{def:rankwidth},
the {\em signed rank-width} of a signed graph $G$
as the branch-width of the signed cut-rank function
$\rho^\pm_G(U)=\rho_{G^+}(U)+\rho_{G^-}(U)$.

\begin{definition}[Rank-width of formulas]\rm
\label{def:sigrwf}
The {\em(signed) rank-width} $\rwd(\phi)$ of a CNF formula $\phi$
is the signed rank-width of the signed formula graph~$F_\phi$.
\end{definition}

Although our signed rank-width is essentially equivalent to an existing concept of
bi-rank-width of directed graphs as introduced by Kant\'e~\cite{Kan08}
(in the bipartite case, at least),
the latter concept is not widely known and its introduction in the context
of CNF formulas would bring only additional technical complications.
In the SAT context,
it is more natural and easier to deal with undirected signed graphs.
Hence we introduce, following previous Definition~\ref{df:parsetree},
{\em $(t^+,t^-)$-labeling parse trees} which will be equivalent to
signed rank-width (up to a factor of $2$, see Theorem~\ref{thm:signedrwalg}) 
in a way analogical to Theorem~\ref{thm:parse}.

\begin{definition}\rm
\label{def:sigparse}
A {\em$(t^+,t^-)$-labeling parse tree} $T=(T^+,T^-)$ of a signed graph $G$ 
is a pair $(T^+,T^-)$ of two labeling parse trees $T^+$ and $T^-$ such that: 
\begin{enumerate}[I.]\vspace*{-3pt}
\item 
$T^+$ ($T^-$) is a $t^+$-labeling ($t^-$-labeling) parse tree generating
$G^+$ ($G^-$), and 
\item
The underlying rooted ordered trees of $T^+$ and $T^-$ are identical.
\end{enumerate}
With a slight abuse of terminology, we will refer to the pair of subtrees
of $T^+$ and $T^-$ rooted at a common node $s$ as to a subtree of~$T$
rooted at~$s$.
\end{definition}

Analogically to labeled graphs of Section~\ref{sub:parsetr}, 
\nb{consider moving this paragraph to previous section -- keep here for the
long version, the abstract will skip parse tree stuff}
we call a signed graph $G$ with associated pair of labelings
$lab^+:V(G)\to2^{[t^+]}$ and $lab^-:V(G)\to2^{[t^-]}$
a {\em$(t^+,t^-)$-labeled graph} $\dbar G=(G,lab^+,lab^-)$.
We shortly refer to the $t^+$-labeled graph $(G^+,lab^+)$ as to
$\dbar G^+$, and analogically to $\dbar G^-$.
The scope of the join operation $\ljoin$ (Definition~\ref{def:ljoin})
can then be extended in a natural way as
$$
\dbar G_1\ljoin\dbar G_2\>=\>
	\big(\dbar G_1^+ \ljoin \dbar G_2^+\big)\cup
	\big(\dbar G_1^- \ljoin \dbar G_2^-\big)
.$$

In our paper we propose signed rank-width as a way of measuring
complexity of formulas  that fares significantly better than previously
considered signed clique-width of $F_\phi$ (e.g.~\cite{FMR07}).
{\em Signed clique-width} is the natural extension of clique-width having
two separate operators for creating the `plus' and the `minus' edges.
The advantage of our approach is witnessed by the following two claims.
\nb{mention, that for clique-width the colours remain the same
 [similarly to bi-rank-width] -- no, this is now different...}

\begin{proposition}
\label{pro:srwinequalities}
Let $\phi$ be an arbitrary CNF formula and $F_\phi$ its signed graph
of signed clique-width~$\cwd(\phi)$.
Then the following are true
\begin{enumerate}[a)]\vspace*{-3pt}
\item
$\rwd(\phi)\leq2\cwd(\phi)$,
\item
there exist instances $\phi$ such that $\cwd(\phi)\geq2^{\rwd(\phi)/4-1}$.
\end{enumerate}
\end{proposition}

\begin{proof}
a)\nb{correction of wrong proof}
Assume a signed $k$-expression tree $S$ for $\phi$ where $k=\cwd(\phi)$.
Clearly, $S$ gives ordinary $k$-expression trees for each of $F_\phi^+$, $F_\phi^-$.
Now, analogically to Theorem~\ref{thm:rwinequalities}(a) \cite{OS06},
the rank-decompositions of $F_\phi^+$ and $F_\phi^-$ with the same
underlying tree as $S$ have widths $\leq k$ each,
and hence $\rwd(\phi)\leq k+k=2\cwd(\phi)$.

\smallskip
b) We define $\phi$ such that $F_\phi^+=G$ where,
cf.\ Theorem~\ref{thm:rwinequalities}(b),
$\cwd(G)\geq2^{\rwd(G)/2-1}$,
and $F_\phi^-$ is arbitrary such that its rank-decomposition inherited from
that of $G$ has width $\leq\rwd(G)$.
Then $\rwd(\phi)\leq2\rwd(G)$ and the claim follows since
$\cwd(\phi)\geq\cwd(G)$.
\qed\end{proof}

\begin{theorem}[\cite{HO08}]
\label{thm:signedrwalg}~
There is an FPT algorithm that, for a fixed parameter $t$ and a given CNF
formula $\phi$,
either finds a $(t^+,t^-)$-labeling parse tree for the formula graph
$F_{\phi}$ where $t^+\!\leq t$ and $t^-\!\leq t$,
or confirms that the signed rank-width of $\phi$ is more than~$t$.
\end{theorem}

\begin{proof}
We show that the algorithm of \cite{HO08} (Theorem~\ref{thm:rwalg})
can be used to compute also the signed rank-width $\leq t$ of a signed graph
$G=F_\phi$.
Indeed, we define a new graph $G'$ as the union of $G^+$ and a vertex-disjoint
copy $G_0^-$ of $G^-$, and a partition $\ca P$ of $V(G')$ as the collection of the
corresponding vertex pairs from $V(G^+)\times V(G_0^-)$.
Then we call the algorithm of \cite{HO08} to either compute a {\em $\ca P$-partitioned}
rank-decomposition of $G'$ of width $\leq t$,
or confirm that the $\ca P$-partitioned width is~$>\!t$.
This width is exactly our signed rank-width of $G$ since,
for any bipartition $(U,W)$ of $V(G')$ not crossing $\ca P$,
the cut-rank $\rho_{G'}(U)$ of $U$ in $G'$ trivially equals our signed cut-rank 
$\rho^\pm_G(X)=\rho_{G^+}(X)+\rho_{G^-}(X)$
of $X=U\cap V(G)$.
Lastly, applying Theorem~\ref{thm:parse}, we separately transform
the two ``inherited'' rank-decompositions of $G^+$ and $G_0^-$
into parse trees $T^+$ and $T^-$,
and output together the $(t^+,t^-)$-labeling parse tree $T=(T^+,T^-)$ of
$F_\phi$.
\qed\end{proof}

\section{Algorithm for propositional model counting \numsat}
\label{sec:numsat}

This section proves our most important result -- the \numsat part of
Theorem~\ref{thm:main}.

We remind the readers that the previous best algorithm
\cite{FMR07} for \numsat\ on graphs of bounded clique-width
has had a single-exponential runtime dependency on the signed
clique-width of a formula (and this dependency cannot be further improved
unless the so called exponential time hypothesis fails).
Hence by Proposition~\ref{pro:srwinequalities}.b the worst case scenario
for the algorithm of \cite{FMR07} would lead to a double-exponential
runtime dependency on the signed rank-width of the formula.
On the other hand:

\begin{theorem}
\label{thm:nsat}
Given a CNF formula $\phi$ and
a $(t^+,t^-)$-labeling parse tree (Theorem~\ref{thm:signedrwalg}) 
of the formula graph $F_\phi$ (Definition~\ref{def:sigparse}),
there is an algorithm that
counts the number of satisfying assignments of $\phi$ in time
$$\ca O( t^3\cdot 2^{3t(t+1)/2}\cdot|\phi|)
	\qquad\mbox{where $t=\max(t^+,t^-)$.}
$$
\end{theorem}

\subsection{Informal notes}

Our algorithm (see Algorithm~\ref{alg:nsat}) proving Theorem~\ref{thm:nsat}
applies the dynamic programming paradigm on the parse trees of the formula
graph $F_\phi$ (constructed by Theorem~\ref{thm:signedrwalg}).
This is, on one hand, a standard approach utilized also by
Fischer, Makowsky and Ravve \cite{FMR07}.
On the other hand, however, comparing to \cite{FMR07}
we achieve an exponential runtime speedup in terms of rank-width.
This significant improvement has two main sources (necessary on both sides):
\begin{itemize}
\item
We heavily apply the basic calculus and tools of linear algebra in the algorithm
(which is indeed natural in view of the algebraic definition of rank-width).
See the details in Subsection~\ref{sub:linalg}.
\item
Our dynamic programming algorithm is built upon the idea of an
``expectation'' (when processing the parse tree of the input) --
in addition to the information recorded about a partial solution
processed so far,
we also record what is expected from a complementary partial solution
coming from the unprocessed part of the input.
\end{itemize}

Especially the second point deserves an informal explanation before giving
a formal description in Definition~\ref{def:signature}.\,II.
The background idea is that the amount of information one has to remember
about a partial solution shrinks a lot if one ``knows'' what the complete
solution will look like.
Such saving sometimes largely exceeds the cost of keeping an exhaustive list
of all possible future ``shapes'' of complete solutions.
This is also our case where the application is quite natural --
we may exhaustively preprocess the values of some variables in advance.

The idea of using an ``expectation'' to speed up a dynamic
programming algorithm on a rank-decomposition
has first appeared in Bui-Xuan, Telle and Vatshelle~\cite{BTV10}
in relation to solving the dominating set on graphs of bounded rank-width.
This concept has been subsequently formalized and generalized
by the authors in \cite{GH09} (in the so called PCE scheme formalism).
Furthermore, it has also been shown \cite[Proposition~5.1]{GH09} that
use of the ``expectation'' concept is unavoidable
to achieve speed up for the dominating set problem.

Unfortunately, we cannot simply refer the formalism of \cite{GH09} here
since it was designed for optimization, and not enumeration, problems.
We thus have to describe it again from scratch in our Algorithm~\ref{alg:nsat}.

\subsection{Supplementary technical concepts}
\label{sub:linalg}

This part describes several technical concepts 
needed to formulate all details of coming Algorithm~\ref{alg:nsat}.
It may be skipped during the first reading.

A useful algebraic concept is that of orthogonality. 
We say that labeling $\ell$ is
{\em orthogonal} to a set of labelings $X$
if $\ell$ has even intersection with every element of $X$
(i.e.\ the scalar product of the labeling vectors is $0$ over $\GFii$). 
Remember that for $t$-labeling parse trees, 
in order for two vertices become adjacent by the join operation $\ljoin$, 
their labelings need to have odd intersection, i.e.\ to be non-orthogonal.
The power of orthogonality comes from the following rather trivial claim 
occurring already in \cite{BTV10,GH09}:

\begin{lemma}
\label{lem:ortspace}
Assume $t$-labeled graphs $\bar G$ and $\bar H$,
and arbitrary $X\subseteq V(\bar G)$ and $y\in V(\bar H)$.
In the join graph $\bar G\pplusoid\bar H$, %
the vertex $y$ is adjacent to some vertex in $X$
if and only if the vector subspace spanned by the $\bar G$-labelings of
the vertices of $X$
is not orthogonal to the $\bar H$-labeling vector of~$y$ in $\GFii^t$.
\end{lemma}

In view of Lemma~\ref{lem:ortspace},
the following result will be useful in deriving the complexity of our algorithm.
\begin{lemma}[\cite{GR69}, cf.~{\cite[Proposition~6.1]{GH09}}]
\label{lem:subnum}
\smallskip
The number $S(t)$ of subspaces of the binary vector space $\GFii^t$
satisfies $S(t)\leq 2^{t(t+1)/4}$ for all $t\geq12$.
\end{lemma}

We recall from Definition~\ref{def:formgraph} the signed graph $F_\phi$ of a
formula $\phi$ on a vertex set $V(F_\phi)=W\cup C$ where $W$ is the set of
variables and $C$ is the set of clauses of $\phi$. An \emph{assignment} is
then a mapping $\nu: W\to\{0,1\}$. In the course of computation of our
algorithm we will need to remember some local information about all satisfying
assignments for $\phi$. The information to be remembered for each such
assignment is formally described by the following definition.

\begin{definition}\rm
\label{def:signature}
Consider an arbitrary $(t^+,t^-)$-labeling $\dbar F_1=(F_1,lab^+,lab^-)$
of a
signed subgraph $F_1\subseteq F_\phi$, 
and any partial assignment $\nu_1: V(F_1)\cap W\to\{0,1\}$.
We say that $\nu_1$ is an {\em assignment of shape}
$(\Sigma^+,\Sigma^-,\Pi^+,\Pi^-)$ in $\dbar F_1$ if
\begin{enumerate}[I.]\vspace*{-3pt}
\item
$\Sigma^+$ is the subspace of $\GFii^t$ generated by the label vectors
$lab^+(\nu_1^{-1}(1))$ and
$\Sigma^-$ is the subspace of $\GFii^t$ generated by $lab^-(\nu_1^{-1}(0))$,
and
\item
$\Pi^+,\Pi^-$ are subspaces of $\GFii^t$ such that,
for every clause $c\in V(F_1)\cap C$,
at least one of the following is true
\begin{itemize}
\item
$c$ is adjacent to some vertex from $\nu_1^{-1}(1)$ in $F_1^+$
or to some vertex from $\nu_1^{-1}(0)$ in $F_1^-$, or
\item
the label vector $lab^+(c)$ is not orthogonal to $\Pi^+$
or $lab^-(c)$ is not orthogonal to~$\Pi^-$
(cf.~Lemma~\ref{lem:ortspace}).
\end{itemize}
\end{enumerate}
\end{definition}

Very informally saying, I. states which true literals in $F_1$ (w.r.t.\
$\nu_1$) are available to satisfy clauses of $F_\phi$, and II. stipulates
that every clause in $F_1$ is satisfied by a true literal in $F_1$ or is
expected to be satisfied by some literal in $F_\phi-V(F_1)$. \nb{mention
  the difference between $\Pi^+,\Pi^-$?} Note that one partial
assignment $\nu_1$ could be of several distinct shapes, which differ in
$\Pi^+,\Pi^-$. (This is true even for complete assignments.) Moreover, there
is no requirement on $\Pi^+$ and $\Pi^-$ to have an empty intersection with
$\Sigma^+$, $\Sigma^-$ and each other. 
The trivial useful properties of assignment shapes are:

\begin{proposition}
\label{pro:signature}
We consider a CNF formula $\phi$ with the variable set $W$,
and any assignment $\nu:W\to\{0,1\}$.
Assume $\dbar F_1,\dbar F_2$ are $(t^+,t^-)$-labeled graphs such that
$F_\phi=\dbar F_1\ljoin\dbar F_2$,
and let $\nu_1,\nu_2$ denote the restrictions of $\nu$ to $\dbar F_1,\dbar
F_2$.
\begin{enumerate}[a)]\vspace*{-3pt}
\item
The assignment $\nu$ is satisfying for~$\phi$ if, and only if,
there exist subspaces $\Sigma^+,\Sigma^-,\Pi^+,\Pi^-$ of $\GFii^t$ 
such that $\nu_1$ is of shape $(\Sigma^+,\Sigma^-,\Pi^+,\Pi^-)$ 
in $\dbar F_1$
and $\nu_2$ is of shape $(\Pi^+,\Pi^-,\Sigma^+,\Sigma^-)$ in~$\dbar F_2$.
\item 
If, in $\dbar F_1$, $\nu_1$ is of shape
$(\Sigma_0^+,\Sigma_0^-,\Pi_0^+,\Pi_0^-)$ and, at the same time, $\nu_1$ is of shape
$(\Sigma_1^+,\Sigma_1^-,\Pi_1^+,\Pi_1^-)$, then $\Sigma_0^+=\Sigma_1^+$ and
$\Sigma_0^-=\Sigma_1^-$.
\item 
The assignment $\nu_1$ is satisfying for~$\phi_1$
-- the subformula of $\phi$ represented by~$F_1$ if, and only if,
$\nu_1$ is of shape $(\Sigma^+,\Sigma^-,\emptyset,\emptyset)$
for some subspaces $\Sigma^+,\Sigma^-$.
\end{enumerate}
\qed\end{proposition}

\subsection{The dynamic processing algorithm}

We now return to our Theorem~\ref{thm:nsat},
considering a $(t^+,t^-)$-labeling parse tree $T_\phi$ of a given formula
graph $F_\phi$.
The core of our bottom-up dynamic processing of $T_\phi$ is as follows:
At every node $z$ such that the subtree of $T_\phi$ rooted at $z$
parses a $(t^+,t^-)$-labeled graph $\dbar F_z$,
we record an integer-valued array $\atable_z$ indexed by
all the quadruples of subspaces of $\GFii^t$, where $t=\max(t^+,t^-)$.
The value of the entry $\atable_z[\Sigma^+,\Sigma^-,\Pi^+,\Pi^-]$ is equal to
the number of variable assignments in $F_z\subseteq F_\phi$
that are of the shape $(\Sigma^+,\Sigma^-,\Pi^+,\Pi^-)$ in $\dbar F_z$
(cf.~Definition~\ref{def:signature}).

For a subset $X\subseteq\GFii^t$, let $\langle X\rangle$ denote the vector
subspace of $\GFii^t$ spanned by the points of~$X$.
If $f$ is a relabeling, i.e.\ a linear transformation defined by a binary
matrix $\mx R_f$, then $f(X)$ denotes the image of $X$ under $f$,
and $f^T(X)$ denotes the image of $X$ under the {\em transposed} 
relabeling given by $\mx R_f^T$.

\begin{algorithm}[Theorem~\ref{thm:nsat}]
\label{alg:nsat}
Given is a CNF formula $\phi$ and
a signed $(t^+,t^-)$-labeling parse tree $T_\phi$ of the formula graph $F_\phi$.
\begin{enumerate}[1.]
\item
We initialize all entries of $\atable_z$ for $z\in V(T_\phi)$ to~$0$.
\smallskip\item
We process all nodes of $T_\phi$ in the leaves-to-root order as follows.

\smallskip
\begin{enumerate}[a)]\parskip1pt
\item\label{it:alg2a}
At a clause leaf $c$ of $T_\phi$,
we set ~$\atable_c[\emptyset,\emptyset,\Pi^+,\Pi^-]\leftarrow1$~
for all subspaces $\Pi^+,\Pi^-$ such that at least one of them is not
orthogonal to the label vector of $\{1\}$ 
(and $\emptyset$ stands for the zero subspace).
\item\label{it:alg2b}
At a variable leaf $\ell$ of $T_\phi$,
we set ~$\atable_\ell[\langle\{1\}\rangle,\emptyset,\Pi^+,\Pi^-]\leftarrow1$~
and ~$\atable_\ell[\emptyset,\langle\{1\}\rangle,\Pi^+,\Pi^-]\leftarrow1$~
for all pairs $\Pi^+,\Pi^-$.
\smallskip
\item\label{it:alg2c}
Consider an internal node $z$ of $T_\phi$, with the left son $x$ and the
right son $y$ such that $\atable_x$ and $\atable_y$ have already been
computed.

\begin{itemize}
\item
Let the composition operators at $z$ in the labeling parse trees
$T_\phi^+,T_\phi^-$ be as 
$\dbar F_z^+=\dbar F_x^+\pplus{g^+}{f_1^+,f_2^+}\dbar F_y^+$
and $\dbar F_z^-=\dbar F_x^-\pplus{g^-}{f_1^-,f_2^-}\dbar F_y^-$
(cf.\ Definition~\ref{def:sigparse} for $T_\phi$).

\item
We loop exhaustively over all indices to $\atable_x,\atable_y,\atable_z$,
i.e. over all $12$-tuples of subspaces
$\Sigma_x^+,\Sigma_x^-,\Pi_x^+,\Pi_x^-$,
$\Sigma_y^+,\Sigma_y^-,\Pi_y^+,\Pi_y^-$,
$\Sigma_z^+,\Sigma_z^-,\Pi_z^+,\Pi_z^-$ of $\GFii^t$.
If all the following are true
\begin{itemize}\vskip2pt
\item[]
$\Sigma_z^+=\big\langle f_1^+(\Sigma_x^+)\cup f_2^+(\Sigma_y^+)\big\rangle$
and
$\Sigma_z^-=\big\langle f_1^-(\Sigma_x^-)\cup f_2^-(\Sigma_y^-)\big\rangle$,
\item[]
$\Pi_x^+=\big\langle f_1^+{}^T(\Pi_z^+)\cup g^+(\Sigma_y^+)\big\rangle$
and
$\Pi_x^-=\big\langle f_1^-{}^T(\Pi_z^-)\cup g^-(\Sigma_y^-)\big\rangle$,
\item[]
$\Pi_y^+=\big\langle f_2^+{}^T(\Pi_z^+)\cup g^+{}^T(\Sigma_x^+)\big\rangle$
and
$\Pi_y^-=\big\langle f_2^-{}^T(\Pi_z^-)\cup g^-{}^T(\Sigma_x^-)\big\rangle$,
\end{itemize}\vskip2pt

then we add the product 
$\atable_x[\Sigma_x^+,\Sigma_x^-,\Pi_x^+,\Pi_x^-]\cdot
	\atable_y[\Sigma_y^+,\Sigma_y^-,$ $\Pi_y^+,\Pi_y^-]$
to the table entry $\atable_z[\Sigma_z^+,\Sigma_z^-,$ $\Pi_z^+,\Pi_z^-]$.
\end{itemize}
\end{enumerate}

\smallskip\item\label{it:alg3}
We sum up all the entries $\atable_r[\Sigma^+,\Sigma^-,\emptyset,\emptyset]$
where $r$ is the root of $T_\phi$ and $\Sigma^+,\Sigma^-$ are arbitrary
subspaces of $\GFii^t$.
This is the resulting number of satisfying assignments of $\phi$.
\end{enumerate}
\end{algorithm}

\begin{proof}[Algorithm~\ref{alg:nsat}\,/\,Theorem~\ref{thm:nsat}]
The task is to prove that
the computed value $\atable_z[\Sigma^+,\Sigma^-,\Pi^+,\Pi^-]$ 
is indeed equal to the number of assignments in $\dbar F_z$
that are of the shape $(\Sigma^+,\Sigma^-,\Pi^+,\Pi^-)$.
This is done by structural induction on $z$ ranging 
from the leaves of $T_\phi$ to its root.
Then, in step \ref{it:alg3} of our algorithm,
the computed number of satisfying assignments of $\phi$ is correct
by Proposition~\ref{pro:signature}\,b,c.

\begin{description}
\item[\ref{it:alg2a}.]
If $z=c$ where $c$ is a clause leaf,
then $F_c$ defines a formula with one empty (so far false) clause $c$.
There is only one possible assignment in $\dbar F_c$.
In order to satisfy $c$, its labeling $lab^+(c)=lab^-(c)=\{1\}$
should not be orthogonal to expected $\Pi^+$ or $\Pi^-$
(Definition~\ref{def:signature}.\,II),
as done in step \ref{it:alg2a}.

\item[\ref{it:alg2b}.]
If $z=\ell$ where $\ell$ is a variable leaf,
then $F_\ell$ defines a formula with one variable and no clause.
There are two assignments of $\ell$ and no requirement on $\Pi^+,\Pi^-$ from
Definition~\ref{def:signature}.\,II.
Hence these two assignments contribute $1$ each to all the indicated table entries
by Definition~\ref{def:signature}.\,I.

\medskip
\item[\ref{it:alg2c}.]
This is the hard core of our proof.
By induction both $\atable_x,\atable_y$ already contain the correct values.
Assume we have a partial assignment $\nu$ in $\dbar F_z$ of shape
$(\Sigma_z^+,\Sigma_z^-,\Pi_z^+,\Pi_z^-)$.
Then $\nu$ defines partial assignments $\nu_x$ in $\dbar F_x$ 
and $\nu_y$ in $\dbar F_y$ which, in turn,
uniquely determine the corresponding subspaces $\Sigma_x^+,\Sigma_x^-$ and
$\Sigma_y^+,\Sigma_y^-$ by Proposition~\ref{pro:signature}\,b).
It follows from Definition~\ref{def:signature}.\,II of a shape 
and from Definition~\ref{df:parsetree} of the composition operators
$\pplus{g^+}{f_1^+,f_2^+}$ and $\pplus{g^-}{f_1^-,f_2^-}$ at $z$,
that $\nu_x$ is of shape $(\Sigma_x^+,\Sigma_x^-,\Pi_x^+,\Pi_x^-)$
 for some $\Pi_x^+,\Pi_x^-$ if
\begin{enumerate}[i.]\smallskip
\item
for every clause $c$ not adjacent to $\nu_x^{-1}(1)$ in $F_x^+$,
nor to $\nu_x^{-1}(0)$ in $F_x^-$,
we have that $f_1^+(lab_x^+(c))$ is not orthogonal to $\Pi_z^+$
or $f_1^-(lab_x^-(c))$ is not orthogonal to $\Pi_z^-$
(informally, $c$ will be satisfied by the expectation at $z$ after
relabeling),
\item
or $lab_x^+(c)$ is not orthogonal to $g^+(\Sigma_y^+)$
or $lab_x^-(c)$ is not orthogonal to $g^-(\Sigma_y^-)$
(informally, $c$ is satisfied by a true literal coming from $\dbar F_y$ 
in the labeling composition at $z$).
\end{enumerate}\smallskip

Note that, e.g., $f_1^+(lab_x^+(c))$ is not orthogonal to 
$\mx v\in\Pi_z^+$ iff 
$$1=f_1^+(lab_x^+(c))\times \mx v^T=
	\left(lab_x^+(c)\times \mx R_{f_1^+}\right)\times \mx v^T=
	lab_x^+(c)\times \left(\mx v\times\mx R_{f_1^+}^{\>T}\right)^T
.$$

Hence putting the two disjoint alternatives i,ii for $c$ together,
we see that $lab_x^+(c)$ should not be orthogonal to 
$f_1^+{}^T(\Pi_z^+)\cup g^+(\Sigma_y^+)$
or $lab_x^-(c)$ should not be orthogonal to 
$f_1^-{}^T(\Pi_z^-)\cup g^-(\Sigma_y^-)$.
This exactly corresponds to the condition on $\Pi_x^+$ and $\Pi_x^-$
in \ref{it:alg2c}.
Analogical fact is true for $\Pi_y^+$ and $\Pi_y^-$.
Therefore, $\nu_x$ and $\nu_y$ have been accounted for in
$\atable_x$ and $\atable_y$, respectively,
and so $\nu$ is now counted in $\atable_z[\Sigma_z^+,\Sigma_z^-,$
$\Pi_z^+,\Pi_z^-]$.

\medskip
On the other hand, we have to prove that no assignment is counted
more than once in one particular entry 
$\atable_z[\Sigma_z^+,\Sigma_z^-,\Pi_z^+,\Pi_z^-]$.
This is not immediate due to a (limited) freedom in a choice of
the ``expectation'' part of shape in the previous arguments.
For $\atable_x,\atable_y$ this claim is true by induction.
Any particular partial assignment $\nu$ in $\dbar F_z$
uniquely determines $\Sigma_z^+,\Sigma_z^-$, and
$\Sigma_x^+,\Sigma_x^-$, $\Sigma_y^+,\Sigma_y^-$ as above.
Then the conditions in \ref{it:alg2c} of the algorithm
also uniquely determine $\Pi_x^+,\Pi_x^-$ and $\Pi_y^+,\Pi_y^-$
(and so their entries in $\atable_x,\atable_y$).
Hence the assignment $\nu$ is counted at most once for every
particular choice of $\Pi_z^+,\Pi_z^-$, too.
\end{description}

Lastly, we analyze the runtime of our algorithm.
Let $S(t)$ be the number of subspaces of $\GFii^t$.
Every single call to one of the steps 1, \ref{it:alg2a},
\ref{it:alg2b}, and \ref{it:alg3} of Algorithm~\ref{alg:nsat}
is proportional to the size of the table which is $\ca O\big(S(t)^4\big)$.
One call to \ref{it:alg2c} in this algorithm
actually has to loop over all $6$-tuples
$\Sigma_x^+,\Sigma_x^-$, $\Sigma_y^+,\Sigma_y^-$,
$\Pi_z^+,\Pi_z^-$ of subspaces of $\GFii^t$,
while the remaining $6$ subspaces 
$\Pi_x^+,\Pi_x^-$, $\Pi_y^+,\Pi_y^-$, $\Sigma_z^+,\Sigma_z^-$
can be computed in time $\ca O(t^3)$ each using standard algorithms
of linear algebra.
Hence this point takes time $\ca O\big(t^3\cdot S(t)^6\big)$.

For the sake of completeness, we note that there exists
\cite[Lemma~6.3]{GH09} an efficient indexing scheme for all
the subspaces of $\GFii^t$ with query time $\ca O(t^3)$.
Such a scheme can be built in time
$\ca O\big(2^{3t(t+1)/4}\cdot t^3\big)$.

Altogether, using Lemma~\ref{lem:subnum}, our Algorithm~\ref{alg:nsat} takes time
$$
\ca O\big(|V(T_\phi)|\cdot t^3\cdot S(t)^6\big) =
  \ca O\big(|V(T_\phi)|\cdot t^3\cdot 2^{6t(t+1)/4}\big) =
  \ca O\big(|\phi|\cdot t^3\cdot 2^{3t(t+1)/2}\big)
.~~\qed$$
\end{proof}

\section{Algorithm for the \maxsat problem}
\label{sec:maxsat}

The same ideas as presented in Section~\ref{sec:numsat} lead also to a
parameterized algorithm for the \maxsat optimization problem
which asks for the maximum number of satisfied clauses in a CNF formula.
We briefly describe this extension,
though we have to admit that the importance of the \maxsat algorithm 
on graphs of bounded rank-width is not as high as that of \numsat.
The reason for lower applicability is that for ``sparse'' formula graphs 
(i.e.\ those not containing large bipartite cliques) 
their rank-width is bounded iff their tree-width is bounded, 
while for dense formula graphs the satisfiability problem is easier
in general.

\begin{theorem}
\label{thm:maxsat}
There is an algorithm that, given a CNF formula $\phi$ and
a $(t^+,t^-)$-labeling parse tree of the formula graph $F_\phi$,
\smallskip
solves the \maxsat optimization problem of $\phi$ in time
\mbox{$\ca O( t^3\cdot 2^{3t(t+1)/2}\cdot|\phi|)$}
where $t=\max(t^+,t^-)$.
\end{theorem}

In order to formulate this algorithm, we extend
Definition~\ref{def:signature} as follows.
Recall $V(F_\phi)=W\cup C$ where $W$ are the
variables and $C$ are the clauses of~$\phi$.
\begin{definition}\rm
Consider a $(t^+,t^-)$-labeling $\dbar F_1=(F_1,lab^+,lab^-)$ of a
signed subgraph $F_1\subseteq F_\phi$, 
and a partial assignment $\nu_1: V(F_1)\cap W\to\{0,1\}$.
We say that $\nu_1$ is an {\em assignment of defective shape}
$(\Sigma^+,\Sigma^-,\Pi^+,\Pi^-)$ in $\dbar F_1$ if
there exists a set $C_0\subseteq C\cap V(F_1)$ such that
$\nu_1$ is of shape $(\Sigma^+,\Sigma^-,\Pi^+,\Pi^-)$ in $\dbar F_1-C_0$.
The value (the {\em defect}) of $\nu_1$ with respect to this defective shape
is the minimum cardinality of such~$C_0$.
\end{definition}
Informally, the defect equals the number of clauses in $F_1$
which are unsatisfied there and not expected to be satisfied in
a complete assignment in $F_\phi$.

We process the parse tree $T_\phi$ of $F_\phi$ similarly to Algorithm~\ref{alg:nsat},
but this time the value of the entry
$\atable_z[\Sigma^+,\Sigma^-,\Pi^+,\Pi^-]$ will be equal to
the minimum defect over all partial assignments in $\dbar F_z$
that are of defective shape $(\Sigma^+,\Sigma^-,\Pi^+,\Pi^-)$.
Formally:

\begin{algorithm}[Theorem~\ref{thm:maxsat}]
\label{alg:maxsat}
Given is a CNF formula $\phi$ and
a signed $(t^+,t^-)$-labeling parse tree $T_\phi$ of the formula graph $F_\phi$.
\begin{enumerate}[1.]
\item
We initialize all entries of $\atable_z$ for $z\in V(T_\phi)$ to~$\infty$.
\smallskip\item
We process all nodes of $T_\phi$ in the leaves-to-root order as follows.

\smallskip
\begin{enumerate}[a)]\parskip1pt
\item\label{it:algm2a}
At a clause leaf $c$ of $T_\phi$,
we set ~$\atable_c[\emptyset,\emptyset,\Pi^+,\Pi^-]\leftarrow1$~
for all subspaces $\Pi^+,\Pi^-$ that are both orthogonal 
to the label vector of $\{1\}$,
and set ~$\atable_c[\emptyset,\emptyset,\Pi^+,\Pi^-]\leftarrow0$~ otherwise.
\item\label{it:algm2b}
At a variable leaf $\ell$ of $T_\phi$,
we set ~$\atable_\ell[\langle\{1\}\rangle,\emptyset,\Pi^+,\Pi^-]\leftarrow0$~
and ~$\atable_\ell[\emptyset,\langle\{1\}\rangle,\Pi^+,\Pi^-]\leftarrow0$~
for all pairs $\Pi^+,\Pi^-$.
\item\label{it:algm2c}
Consider an internal node $z$ of $T_\phi$, with the left son $x$ and the
right son $y$ such that $\atable_x$ and $\atable_y$ have already been
computed.

\begin{itemize}
\item
Let the composition operators at $z$ in the labeling parse trees
$T_\phi^+,T_\phi^-$ be as 
$\dbar F_z^+=\dbar F_x^+\pplus{g^+}{f_1^+,f_2^+}\dbar F_y^+$
and $\dbar F_z^-=\dbar F_x^-\pplus{g^-}{f_1^-,f_2^-}\dbar F_y^-$.%
\smallskip

\item
We loop exhaustively over all indices to $\atable_x,\atable_y,\atable_z$,
i.e. over all $12$-tuples of subspaces
$\Sigma_x^+,\Sigma_x^-,\Pi_x^+,\Pi_x^-$,
$\Sigma_y^+,\Sigma_y^-,\Pi_y^+,\Pi_y^-$,
$\Sigma_z^+,\Sigma_z^-,\Pi_z^+,\Pi_z^-$ of $\GFii^t$.
If all the following are true
\begin{itemize}\vskip2pt
\item[]
$\Sigma_z^+=\big\langle f_1^+(\Sigma_x^+)\cup f_2^+(\Sigma_y^+)\big\rangle$
and
$\Sigma_z^-=\big\langle f_1^-(\Sigma_x^-)\cup f_2^-(\Sigma_y^-)\big\rangle$,
\item[]
$\Pi_x^+=\big\langle f_1^+{}^T(\Pi_z^+)\cup g^+(\Sigma_y^+)\big\rangle$
and
$\Pi_x^-=\big\langle f_1^-{}^T(\Pi_z^-)\cup g^-(\Sigma_y^-)\big\rangle$,
\item[]
$\Pi_y^+=\big\langle f_2^+{}^T(\Pi_z^+)\cup g^+{}^T(\Sigma_x^+)\big\rangle$
and
$\Pi_y^-=\big\langle f_2^-{}^T(\Pi_z^-)\cup g^-{}^T(\Sigma_x^-)\big\rangle$,
\end{itemize}\vskip2pt

then we let
$m=\atable_x[\Sigma_x^+,\Sigma_x^-,\Pi_x^+,\Pi_x^-]+
	\atable_y[\Sigma_y^+,\Sigma_y^-,\Pi_y^+,\Pi_y^-]$.
If, furthermore, $m<\atable_z[\Sigma_z^+,\Sigma_z^-,\Pi_z^+,\Pi_z^-]$,
then we set
$\atable_z[\Sigma_z^+,\Sigma_z^-,$ $\Pi_z^+,\Pi_z^-]\leftarrow m$.
\end{itemize}
\end{enumerate}

\smallskip\item\label{it:algm3}
We find the minimum $m$ over all the entries 
$\atable_r[\Sigma^+,\Sigma^-,\emptyset,\emptyset]$
where $r$ is the root of $T_\phi$ and $\Sigma^+,\Sigma^-$ are arbitrary
subspaces of $\GFii^t$.
An optimal solution to \maxsat of $\phi$ then has $|C|-m$ satisfied clauses.
\end{enumerate}
\end{algorithm}

\begin{proof}[Algorithm~\ref{alg:maxsat}\,/\,Theorem~\ref{thm:maxsat},
	sketch]
The main task is to show by means of structural induction
that the algorithm correctly computes
in $\atable_z[\Sigma^+,\Sigma^-,\Pi^+,\Pi^-]$
the minimum defect over all partial assignments in $\dbar F_z$
that are of defective shape $(\Sigma^+,\Sigma^-,\Pi^+,\Pi^-)$.

Our proof proceeds in the same way as the proof of Algorithm~\ref{alg:nsat},
with a use of the following easy claim (notation as in the referred proof):
\begin{itemize}
\item
The defect of a partial assignment $\nu$ in $\dbar F_z$ w.r.t.\
$(\Sigma_z^+,\Sigma_z^-,\Pi_z^+,\Pi_z^-)$
equals the sum of the defects of $\nu_x$ in $\dbar F_x$ and $\nu_y$
in $\dbar F_y$ w.r.t.\ defective shapes
$(\Sigma_x^+,\Sigma_x^-,\Pi_x^+,\Pi_x^-)$ and
$(\Sigma_y^+,\Sigma_y^-,\Pi_y^+,\Pi_y^-)$, respectively.
\end{itemize}
The runtime analysis follows, too.
\qed\end{proof}

\section{Conclusions}

We have presented new FPT algorithms for the \numsat and \maxsat problems 
on formulas of bounded rank-width. 
Our algorithms are
single-exponential in rank-width and linear in the size of the formula,
and they do not involve any ``large hidden constants''.
This is a significant improvement over previous results, for several reasons.  
In the case of tree-width this follows from the fact that rank-width is much more
general than tree-width. If a graph has bounded tree-width it also has
bounded rank-width, but there are classes of graphs with arbitrarily high
tree-width and small rank-width (e.g. cliques, complete bipartite graphs,
or distance hereditary graphs). 

As for clique-width (which is bounded iff rank-width is bounded), we have
obtained two
significant improvements over the existing algorithms such as \cite{FMR07}.
Firstly, rank-width can be exponentially smaller than
clique-width, and therefore we obtain an exponential speed-up over the
existing algorithms in the worst case. 
Secondly, there is an FPT algorithm for computing
the rank-width of a graph (and the associated rank-decomposition) exactly, 
whereas in the case of clique-width 
we have to rely on an approximation by an exponential function of rank-width.

Finally, our paper shows that many of the recent ideas and tricks 
of parameterized algorithm design on graphs of bounded rank-width
smoothly translate to certain SAT-related problem instances
which may bring new inspiration to related research, too.

\bibliographystyle{abbrv}
\bibliography{satrw}

\end{document}